\documentclass[12pt]{iopart}
\usepackage{color,cite}
\usepackage{iopams,graphicx}

\makeatletter
\providecommand\underarrow@[3]{%
  \vtop{\ialign{##\crcr$\m@th\hfil#2#3\hfil$\crcr
  \noalign{\nointerlineskip\kern.12\baselineskip}#1#2\crcr}}}
\providecommand{\underrightarrow}{%
  \mathpalette{\underarrow@\rightarrowfill@}}
\providecommand\rightarrowfill@{\arrowfill@\relbar\relbar\longrightarrow}
\providecommand\arrowfill@[4]{%
  $\m@th\thickmuskip0mu\medmuskip\thickmuskip\thinmuskip\thickmuskip
   \relax#4#1\mkern-7mu%
   \cleaders\hbox{$#4\mkern-2mu#2\mkern-2mu$}\hfill
   \mkern-7mu#3$%
}
\makeatother

\newcommand{\erfc}[1]{\mathrm{erfc}\! \left( {#1} \right)}

\newcommand{\bx}{\boldsymbol{\xi}}

\newenvironment{proof}[1][Proof:]{\begin{trivlist}
\item[\hskip \labelsep {\bfseries #1}]}{\end{trivlist}}
\newenvironment{proposition}[1][Proposition:]{\begin{trivlist}
\item[\hskip \labelsep {\bfseries #1}]}{\end{trivlist}}

\newenvironment{remark}[1][Remark:]{\begin{trivlist}
\item[\hskip \labelsep {\bfseries #1}]}{\end{trivlist}}

\newcommand{\qed}{\nobreak \ifvmode \relax \else
      \ifdim\lastskip<1.5em \hskip-\lastskip
      \hskip1.5em plus0em minus0.5em \fi \nobreak
      \vrule height0.75em width0.5em depth0.25em\fi}

\begin{document}

\title[Single file diffusion]{Large deviation function of a tracer position in single file diffusion}
\author{Tridib Sadhu$^{1}$ and Bernard Derrida$^{2}$}
\vspace{1pc}
\address{$^1$ Philippe Meyer Institute for Theoretical Physics, Physics Department, Ecole Normale Sup\'{e}rieure, 24 rue Lhomond, 75231 Paris Cedex 05 - France}
\ead{tridib.sadhu@lpt.ens.fr}
\vspace{1pc}
\address{$^2$ Laboratoire de Physique Statistique, Ecole Normale Sup\'{e}rieure, Universit\'{e} Pierre et Marie Curie,
Universit\'{e} Paris Diderot, CNRS, 24 rue Lhomond, 75231 Paris Cedex 05 - France }
\ead{derrida@lps.ens.fr}
\begin{abstract}
Diffusion of impenetrable particles in a crowded one-dimensional channel is referred as the single file diffusion. The particles do not pass each other and the displacement of each individual particle is sub-diffusive. We analyse a simple realization of this single file diffusion problem where one dimensional Brownian point particles interact only by hard-core repulsion. We show that the large deviation function which characterizes the displacement of a tracer at large time can be computed via a mapping to a problem of non-interacting Brownian particles. We confirm recently obtained results of the one time distribution of the displacement and show how to extend them to the multi-time correlations. The probability distribution of the tracer position depends on whether we take annealed or quenched averages. In the quenched case we notice an exact relation between the distribution of the tracer and the distribution of the current. This relation is in fact much more general and would be valid for arbitrary single file diffusion.  It allows in particular to get the full statistics of the tracer position for the symmetric simple exclusion process (SSEP) at density $1/2$ in the quenched case.
\end{abstract}

\date{\today}
%\pacs{68.35.Fx, 05.50.+q, 68.35.Md, 05.70.Np}
\vspace{2pc}

\submitto{A special issue of J Stat Mech}
\maketitle

\section{Introduction}
Motion in crowded one dimensional channels appears in many physical systems. For example, passage of ions through a narrow pore in a cell membrane \cite{HODGKIN1955}, sliding of proteins along DNA sequences \cite{Li2009}, transport in ionic conductors \cite{Richards1977}, molecular motion inside a porous medium \cite{KARGER1992} or inside carbon nano-tubes \cite{Das2010}, \textit{etc}. What makes the single file diffusion interesting from a theoretical perspective is that the motion of individual particles is sub-diffusive. The variance of the displacement $x(t)$ of a tracer particle over a time $t$ scales as $\sqrt{t}$ as opposed to the linear dependence one finds in normal diffusion. The tracer particle is indistinguishable from the other particles: it only carries a mark. This sub-diffusive scaling has been demonstrated in several experiments on single file systems (transport of water molecules inside carbon nano-tubes \cite{Das2010}, colloidal suspensions in a narrow groove \cite{Wei2000,Lutz2004,Lin2005}, large molecules in a porous medium like Zeolite \cite{KUKLA1996}, \textit{etc}.).

There is a long history of theoretical works on single file diffusion. One of the earliest known results is by Harris \cite{Harris1965} where an explicit formula for the variance was reported for a system of impenetrable Brownian point particles. Later, Levitt \cite{Levitt1973} generalized the work for a system of hard rods and derived an analogous formula for the variance. An interesting feature was noted in \cite{Harris1965,Levitt1973,Jara2006,Arratia1983} that at large times the fluctuations of the tracer position become Gaussian. This was confirmed from an exact solution of the probability distribution of the tracer position for Brownian point particles, which is valid even at finite times \cite{Rodenbeck1998}. The analysis for a system of point particles is simpler compared to a more general single file problem because in the former system the trajectories of the impenetrable particles can be mapped to the trajectories of non-interacting particles \cite{Harris1965}. For more general inter-particle interactions, less results are available: for symmetric simple exclusion process, the variance and the fourth cumulant are known \cite{Arratia1983,KMS2014,KMS_JSP}, while for colloidal systems an expression for the variance was found in terms of isothermal compressibility \cite{Kollmann2003}. An expression of the variance for general single file system was derived in \cite{KMS2014,KMS_JSP} using the macroscopic fluctuation theory \cite{Bertini2014,Jona-Lasinio2010,Bertini2007,Derrida2007}. There is an interesting connection to interface fluctuations where the displacement of a tracer can be mapped to the height of a one-dimensional interface \cite{Majumdar1991}.

In recent years, attention was drawn to an intriguing property of the single file diffusion where the variance depends on the choice of initial state, even at large times \cite{Leibovich2013,KMS2014,KMS_JSP,Gupta2007}. There are two types of initial state one can consider \cite{Gerschenfeld2009}: annealed and quenched. In the annealed case, the starting arrangement of particles are drawn from an equilibrium state of average density $\rho$. This includes configurations where the density profile fluctuates significantly from the average value. In the quenched case, only those initial configurations are considered which at a macroscopic scale have a uniform density profile $\rho$ (for example, a configuration where particles are placed at uniform separation $\rho^{-1}$). It was found \cite{Leibovich2013} that the variance is different in the two settings, although the $\sqrt{t}$ scaling is the same. This is surprising as one would naively expect that a system forgets its initial state after a long time.

Drawing analogy with disordered systems \cite{Gerschenfeld2009}, the two settings can be described in terms of cumulant generating function. For the annealed setting the generating function of the cumulants of the position $x(t)$ of the tagged particle can be defined as
\numparts
\begin{equation}
M_{a}(\lambda,t)=\ln\bigg[\bigg\langle e^{\lambda\, x(t)} \bigg\rangle_{\textrm{\tiny evolution+initial}}\bigg]
\label{eq:mut}
\end{equation}
where $\lambda$ is a fugacity parameter and the average is over both the stochastic evolution and the initial state. On the other hand, for the quenched setting, the same generating function reads
\begin{equation}
M_{q}(\lambda,t)=\bigg\langle\ln\bigg[\bigg\langle e^{\lambda\, x(t)} \bigg\rangle_{\textrm{\tiny evolution}}\bigg]\bigg\rangle_{\textrm{\tiny initial}}
\label{eq:muq}
\end{equation}
\endnumparts
where the average inside logarithm is over the stochastic evolution, whereas the average outside is over the initial state.
Cumulants are obtained as usual by expanding $M_{\alpha}(\lambda,t)$ in powers of $\lambda$:
\begin{equation}
\bigg\langle x(t)^n \bigg\rangle_{\alpha}=\frac{d^nM_{\alpha}(\lambda,t)}{d\lambda^n}\Bigg\vert_{\lambda=0}
\label{eq:cumulant expansion}
\end{equation}
where we use subscript $\alpha$ to denote the annealed and the quenched case. This notation will be used in the rest of the paper.

It was found in \cite{KMS2014} that not only the variance, but all the cumulants of the tracer position scale as $\sqrt{t}$, at large time. This implies that the cumulant generating functions $M_{a}(\lambda,t)$ and $M_{q}(\lambda,t)$ grow as $\sqrt{t}$. This means that for $x(t)\sim \sqrt{t}$, the probability distribution $P_{\alpha}\left(x(t)\right)$ of the tracer position $x(t)$ at time $t$ is of the form
\begin{equation}
P_{\alpha}\left(\xi\sqrt{t}\right)\asymp e^{-\sqrt{t}~\phi_{\alpha}(\xi)}
\label{eq:ldform}
\end{equation}
where $\phi_{\alpha}(\xi)$ is the large deviation function. The symbol $A\asymp B$ denotes that the ratio of the logarithms of $A$ and $B$ goes to $1$ at large $t$, \textit{i.e.}, $\lim_{t\rightarrow \infty}\frac{\ln A}{\ln B}=1$. The large deviation functions are related to their corresponding cumulant generating functions by a Legendre transformation
\begin{equation}
M_{\alpha}(\lambda,t)=\sqrt{t}\>\max_{\xi}\bigg\{\lambda\>\xi-\phi_{\alpha}(\xi)\bigg\}.
\label{eq:legendre}
\end{equation}

Recently, the large deviation function $\phi_{\alpha}(\xi)$ was computed for a single file system of Brownian point particles using the macroscopic fluctuation theory \cite{KMS2014,KMS_JSP} and also starting from microscopic dynamics \cite{Hegde2014,KMS_JSP}. In both approaches the analysis involves rather long calculations.

In this paper, we present an alternative derivation of the large deviation function which in our opinion is simpler. Our analysis is for impenetrable Brownian point particles and relies on a connection between the trajectories of the impenetrable particles and the trajectories of non-interacting particles, \textit{i.e.}, Brownian particles allowed to cross. Using this mapping we calculate the large deviation function in the annealed and the quenched setting, confirming earlier results \cite{KMS2014,KMS_JSP,Hegde2014}. The large deviation functions are different in the two settings; even the asymptotics of the functions differ significantly: in the annealed setting $\phi_{a}(\xi)\sim \vert \xi \vert$, whereas in the quenched setting $\phi_{q}(\xi)\sim \vert \xi\vert^3$, for large $\xi$.

An advantage of our method is that it can be easily extended to analyse the correlations of the tracer position at multiple times. The joint probability $P(x(t_1),x(t_2))$ of the tracer position $x(t_1)$ at time $t_1$ and $x(t_2)$ at time $t_2$ has a form analogous to \eref{eq:ldform}:
\begin{equation}
P_{\alpha}\left(\xi(\tau_1)\sqrt{t},\xi(\tau_2)\sqrt{t}\right)\asymp e^{-\sqrt{t}~\phi_{\alpha}(\xi( \tau_1),\xi(\tau_2))}
\label{eq:ldform two}
\end{equation}
where $t_{j}=\tau_{j}t$ and $x(t_j)=\xi(\tau_j)\sqrt{t}$. We will show how to calculate the large deviation function $\phi_{\alpha}(\xi( \tau_1),\xi(\tau_2))$ for both quenched and annealed cases. This will allow us to recover the two time correlations of the tracer position \cite{Leibovich2013}
\begin{equation}
\langle x(t_{1})x(t_{2})\rangle =\cases{ \kappa(\rho)\bigg[\sqrt{t_1}+\sqrt{t_2}-\sqrt{\vert t_2-t_1\vert }\bigg] & for annealed \cr
\kappa(\rho)\bigg[\sqrt{t_1+t_2}-\sqrt{\vert t_2-t_1 \vert}\bigg] & for quenched }
\label{eq:corr}
\end{equation}
where the pre-factor $\kappa(\rho)$ depends on the density $\rho$, and the variance $\sigma^2/2$ of the displacement of a single particle in unit time when it is diffusing freely in absence of the other particles
\begin{equation}
\kappa(\rho)=\frac{\sigma}{2\rho\sqrt{\pi}}.
\label{eq:curly D}
\end{equation}
The time dependence of the correlation \eref{eq:corr} in the annealed case has the same form as in a fractional Brownian motion \cite{Mandelbrot1968} corresponding to the Hurst exponent $H=1/4$.

The two time correlation function calculated for more general single file problems using the macroscopic fluctuation theory \cite{KMS_JSM} gives the same dependence on time as in \eref{eq:corr} and only the pre-factor $\kappa(\rho)$ depends on the specific details of the system.
Similar time dependence was found in a related model where the order of particles is preserved \cite{Rajesh2001}, and also in height fluctuations of a one-dimensional interface \cite{KrugKallabis1997}.

Our method can further be extended to analyse the joint probability distribution of tracer positions at an arbitrary number of times. By taking a continuous limit this leads to a path integral formulation of the probability of a trajectory of the tracer as a functional $P[x(t)]$ of the trajectory $x(t)$. (We use a notation $[\,]$ to denote a functional.) For an evolution in a time window $[0,t]$ this probability has a form analogous to \eref{eq:ldform} and \eref{eq:ldform two}:
\begin{equation}
P_{\alpha}\left[\xi(\tau)\sqrt{t}\right]\asymp e^{-\sqrt{t}\,\phi_{\alpha}[\xi(\tau)]}
\label{eq:ldform multi}
\end{equation}
where $0\le \tau \le 1$ is the time rescaled by $t$ and $\phi_{\alpha}[\xi(\tau)]$ is the large deviation functional. We shall show that the $\phi_{\alpha}[\xi(\tau)]$ can be formally expressed in terms of the solution of a Schr\"{o}dinger like equation with a moving time dependent step potential.

A direct consequence of \eref{eq:ldform multi} is that the joint probability distribution of the tracer position at multiple times is asymptotically a multivariate Gaussian distribution. This can be seen by a Taylor expansion of $\phi_{\alpha}[\xi(\tau)]$ around the minimum. This implies that all multi-time correlations of the tracer position can be expressed in terms two time correlation \eref{eq:corr} using the Wick's theorem (or the Isserlis' theorem).

Our analysis will be presented in the following order. In \sref{sec:the problem} we define the problem and outline our method of calculating the large deviation function. In \sref{sec:computation} we present an explicit calculation of the large deviation function $\phi_{a}(\xi)$ and $\phi_{q}(\xi)$. In \sref{sec:cumulant one} we use this result to calculate the cumulant generating function and derive expressions of the first few cumulants. In \sref{sec:two time} we extend our method to the two time case, and use this to derive \eref{eq:corr}. A further generalization to multi-time statistics of tracer position is presented in \sref{sec:continuous time}. We conclude in \sref{sec:conclusion} summarizing our results. In the Appendix we derive a relation between the distributions of the position of the tracer particle and of the time-integrated current in the quenched case. This relation is valid for arbitrary single file problem. It allows us to obtain the exact distribution of the displacement of tracer (through $M_{q}(\lambda,t)$) for the symmetric simple exclusion process at density $\rho=1/2$ in the quenched case.

\begin{figure}
\begin{center}
\includegraphics[width=0.5\textwidth]{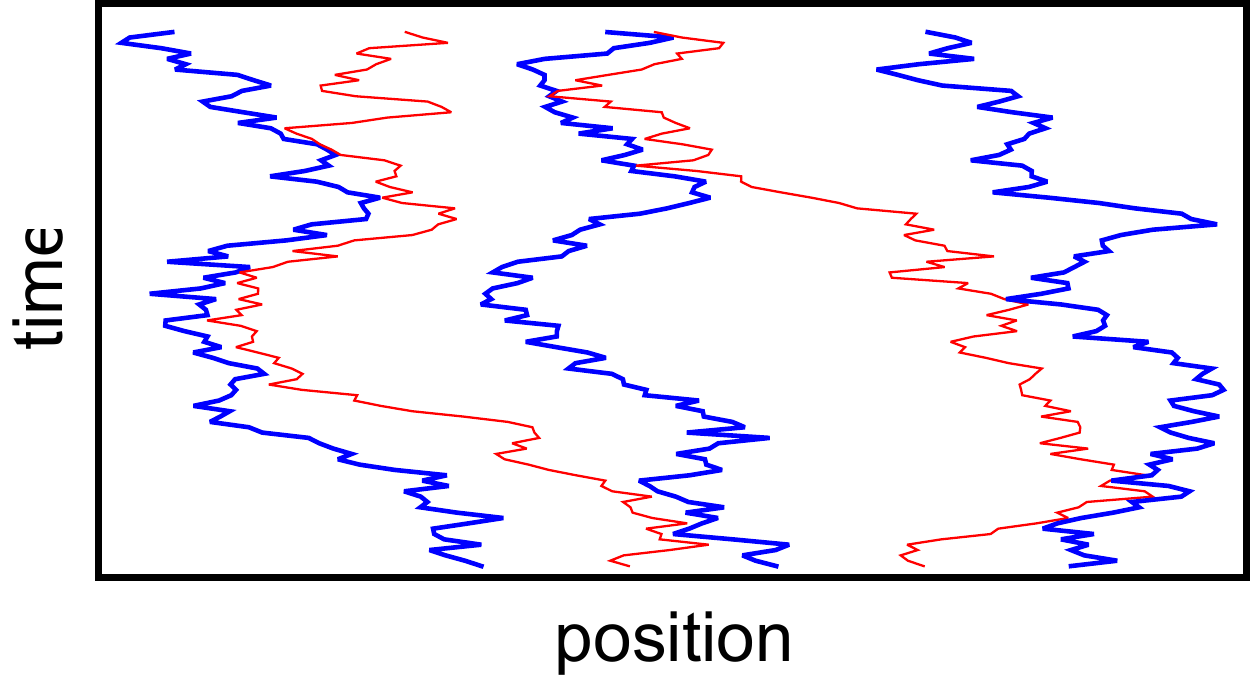}
\caption{A sample trajectory of five Brownian point particles with hard-core repulsion. \label{fig:fig1}}
\end{center}
\end{figure}

\section{The mapping of single file diffusion to non-interacting particles \label{sec:the problem}} We consider a system of Brownian point particles on an infinite line diffusing with hard core repulsion between particles. A sample trajectory of the particles is shown in \fref{fig:fig1}.

At times much smaller than the mean collision time the tracer diffuses normally. With increasing time, the displacement of the tracer is significantly confined by the presence of other particles and the motion becomes sub-diffusive. We analyse the tracer displacement in this large time limit.

\subsubsection*{Mapping to non-interacting particles:}
To derive an expression for the probability \eref{eq:ldform}, we use a mapping to a system of non-interacting particles where trajectories can cross each other. For simplicity we discuss the mapping for a system with a finite number of particles which are ranked as $\{\ldots,-2,-1,0,1,2,\ldots \}$ according to their positions. One can easily generalize the mapping to the case of infinitely many particles. In both the single file and the non-interacting systems we consider the same initial condition, with the zeroth rank particle at $t=0$ being  at the origin. In the single file system we choose the zeroth rank particle as the tracer and define $P(x(t))$ as the probability to find the tracer at position $x(t)$ at time $t$. In the non-interacting system we re-rank the particles at time $t$ according to their positions and define $P_{\textrm{\tiny NI}}(x(t))$ as the probability that the particle of rank zero is at position $x(t)$. These two probabilities are same \cite{Harris1965}:
\begin{equation}
P(x(t))=P_{\textrm{\tiny NI}}(x(t))
\label{eq:id1}
\end{equation}
irrespective of the choice of the initial state as long as it is the same for both the systems.

This equality is simple to argue from the fact that for every history of the non-interacting particles it is always possible to construct a valid history for the impenetrable particles, by simply interchanging the particle labels when two trajectories cross in the former case. As the particles are point particles, this relabelling does not change the probability of the history.

\subsubsection*{The large deviation function:}
In the limit of infinitely many particles distributed with uniform density $\rho$, the probability $P(x(t))$ has the large deviation form \eref{eq:ldform}, and due to \eref{eq:id1} the probability $P_{\textrm{\tiny NI}}(x(t))$ has also the same form. We now show how to calculate the large deviation function $\phi_{\alpha}(\xi)$ by analysing the large time asymptotics of $P_{\textrm{\tiny NI}}(x(t))$. 

To begin with the calculation we define $R_t(x)$ as the number of non-interacting particles which start at any position $\le 0$ and reach  a position $> x$ at time $t$. Similarly, $R_t^\prime(x)$ is the number of non-interacting particles which start at any position $>0$ and reach a position $\le x$ at time $t$. The quantities $R_t(x)$ and $R_t^\prime(x)$ are random variables and depend both on the stochastic evolution and the initial condition. Given that the zeroth rank particle at $t=0$ is at origin, we see that the cumulative probability of the zeroth rank particle at time $t$ to be at a position $x(t)>x$ is equal to the cumulative probability of the difference $R_t(x)-R_t^\prime(x)$ to be strictly positive.
\begin{equation}
	P_{\alpha}(x(t) > x)= P_{\alpha}\bigg[R_t(x)-R_t^\prime(x)\ge 1\bigg]\qquad \textrm{for any $x$,}
	\label{eq:prob R}
\end{equation}
where $\alpha$ denotes the quenched or the annealed initial condition.

As the particles are diffusive it is expected that the probability density of the difference $R_t(x )-R_t^\prime(x)$ has a large deviation form
\begin{equation}
	\mathbb{P}_{\alpha}\left[ \frac{R_t( \xi \sqrt{t})- R_t^\prime(\xi \sqrt{t})}{\sqrt{t}}=r \right]\asymp e^{-\sqrt{t}~\psi_{\alpha}(r,\,\xi)}
	\label{eq:ldform psi}
\end{equation}
where $\psi_{\alpha}(r,\,\xi)$ is the associated large deviation function. Then from \eref{eq:prob R} we are going to show that
\begin{equation}
\phi_{\alpha}(\xi)=\psi_{\alpha}(0,\xi)
	\label{eq:rel 2}
\end{equation}
for positive $\xi$. For negative values of $\xi$ one can use that the $\phi_{\alpha}(\xi)=\phi_{\alpha}(-\xi)$ which is expected as the microscopic dynamics is unbiased. The relation  \eref{eq:rel 2} applies for both settings of initial condition: annealed and quenched.

One way to derive the equality \eref{eq:rel 2} is by taking the limit
\begin{equation}
\lim_{t\rightarrow \infty}\frac{\ln P_{\alpha}(x(t)>\xi\sqrt{t})}{\sqrt{t}}=\lim_{t\rightarrow \infty}\frac{\ln \int_{\xi}^{\infty}dz \, e^{-\sqrt{t}\,\phi_{\alpha}(z)}}{\sqrt{t}}=-\phi_{\alpha}(\xi)
\label{eq:limit1}
\end{equation}
for positive $\xi$ where in the last step we have assumed that for positive $z$ the $\phi_{\alpha}(z)$ is minimum at $z=0$ and is monotonically increasing away from the minimum. Similarly we can show that
\begin{equation}
\lim_{t\rightarrow \infty}\frac{\ln P_{\alpha}\bigg(R_t( \xi \sqrt{t})- R_t^\prime(\xi \sqrt{t})>0\bigg)}{\sqrt{t}}=-\psi_{\alpha}(0,\xi)
\label{eq:limit2}
\end{equation}
where we assumed that $\psi_{\alpha}(r,\xi)$ given in \eref{eq:ldform psi} has minimum at a negative value of $r$ when $\xi$ is positive and the function is monotonically increasing away from the minimum. This is consistent with the property that at a position $x>0$, it is more probable that $R_t(x)$ is less than $R_t^{\prime}(x)$. Combination of the above two equations \eref{eq:limit1} and \eref{eq:limit2} along with \eref{eq:prob R} leads to the result \eref{eq:rel 2} for any positive $\xi$.

Rather than the large deviation function $\psi_{\alpha}(r,\xi)$ it is comparatively easier to calculate its Legendre transformation which gives the cumulant generating function of the difference $R_t( x)- R_t^\prime(x )$. Following the discussion in the introduction, the cumulant generating function can be defined in two settings. In the annealed setting, it is
\numparts
\begin{equation}
	\chi_{a}(B,\xi)=\lim_{t\rightarrow \infty} \frac{1}{\sqrt{t}}\ln \Bigg\langle e^{\>B\left(R_t( \xi \sqrt{t})- R_t^\prime(\xi \sqrt{t} )\right) }\Bigg\rangle_{\textrm{\tiny evolution+initial}}
	\label{eq:chi annealed}
\end{equation}
where $B$ is a fugacity parameter and the average is over all initial configurations of particles and their stochastic evolution. On the other hand in the quenched setting, it is defined as
\begin{equation}
	\chi_{q}(B,\xi)=\lim_{t\rightarrow \infty} \frac{1}{\sqrt{t}}\Bigg\langle\ln \bigg\langle e^{\>B\left(R_t( \xi \sqrt{t})- R_t^\prime(\xi \sqrt{t} )\right)} \bigg\rangle_{\textrm{\tiny evolution}}\Bigg\rangle_{\textrm{\tiny initial}}.
	\label{eq:chi quench}
\end{equation}
\endnumparts

In both settings, the cumulant generating function is related to the corresponding large deviation function by a Legendre transformation
\begin{equation}
	\chi_{\alpha}(B,\xi)=\max_{r}\{B~r-\psi_{\alpha}(r,\xi)\}.
\end{equation}
This leads to a parametric solution of $\psi_{\alpha}(r,\xi)$
\begin{equation}
\psi_{\alpha}(r,\xi)=B\,r-\chi_{\alpha}(B,\xi) \qquad \qquad \textrm{with}\qquad r=\frac{\partial\chi_{\alpha}(B,\xi)}{\partial B}.
\end{equation}
Setting $r=0$ and using the relation \eref{eq:rel 2} we get
\begin{equation}
\fl	\qquad \phi_{\alpha}(\xi)=-\chi_{\alpha}(B,\xi) \qquad \textrm{where $B$ is determined from}\qquad \frac{\partial\chi_{\alpha}(B,\xi)}{\partial B}=0.
	\label{eq:ldf saddle}
\end{equation}

So our approach consists of calculating $\chi_{\alpha}(B,\xi)$ defined in (\ref{eq:chi annealed},\,\ref{eq:chi quench}) and then to use \eref{eq:ldf saddle} to determine the large deviation function $\phi_{\alpha}(\xi)$ of the position of the tracer particle.

\section{Calculation of the large deviation function $\phi_{\alpha}(\xi)$ \label{sec:computation}}
We now calculate $\chi_{a}(B,\xi)$ and $\chi_{q}(B,\xi)$. A key remark is that the difference $R_{t}(x)-R_{t}^{\prime}(x)$ remains unchanged when one maps the single file problem to the non-interacting Brownian particles.
For non-interacting particles located initially at positions $y_j$, we can write using the independence of the particles
\begin{eqnarray}
\fl	\qquad \bigg\langle e^{B\big(R_t(x)-R_t^{\prime}(x)\big)} \bigg\rangle_{\textrm{\tiny evolution}}
&=&\bigg\langle e^{B\>R_t(x)}\bigg\rangle_{\textrm{\tiny evolution}} \bigg\langle e^{-B\>R_t^{\prime}(x)\big)} \bigg\rangle_{\textrm{\tiny evolution}}=\prod_{j}F_t(y_j,B,x)
\label{eq:split one}
\end{eqnarray} 
where the product is over all particles and $F_t(y,B,x)$ is the contribution of a single Brownian particle starting at $y$ given by
\begin{equation}
F_t(y,B,x)=\cases{\bigg\langle e^{B\>\Theta(z(t)-x)}\bigg\rangle_{z(0)=y} & for $y\le 0$\\
\bigg\langle e^{-B\>\Theta(x-z(t))}\bigg\rangle_{z(0)=y} & for $y>0$}
\label{eq:F}
\end{equation}
where $\Theta(x)$ is the Heaviside step function and $z(t)$ is the position of the Brownian particle at time $t$.

In this paper we consider that the initial positions of the particles are distributed according to a uniform density $\rho$.
The annealed and the quenched cases differ by how the average over the initial positions is taken in the generating function \eref{eq:chi annealed} and \eref{eq:chi quench}.
In the annealed case, the average over initial positions is inside the logarithm. We write the average
\begin{equation}
 \Bigg\langle e^{B\big(R_t(x)-R_t^{\prime}(x)\big)} \Bigg\rangle_{\textrm{\tiny evolution+initial}}=\prod_{y}\Bigg[\bigg(1-\rho\>dy\bigg)+\rho\>dy\>F_t(y,B,x)\Bigg]
\end{equation}
where each infinitesimal interval $dy$ is occupied by a single particle with probability $\rho \,dy$ or empty with probability $1-\rho\, dy$. Taking logarithm on both sides leads to
\begin{equation}
\ln \Bigg\langle e^{B\big(R_t(x)-R_t^{\prime}(x)\big)} \Bigg\rangle_{\textrm{\tiny evolution+initial}}=\rho\int_{-\infty}^{\infty}dy\>\bigg[ F_t(y,B,x)-1\bigg].
\end{equation}
On the other hand, in the quenched case, the average over initial positions is outside logarithm. We get
\begin{equation}
\Bigg\langle\ln\bigg\langle e^{B\big(R_t(x)-R_t^{\prime}(x)\big)} \bigg\rangle_{\textrm{\tiny evolution}}\Bigg \rangle_{\textrm{\tiny initial}}=\rho\int_{-\infty}^{\infty}dy\>\ln F_t(y,B,x).
\end{equation}
Substituting this in the formula \eref{eq:chi annealed} and \eref{eq:chi quench} we get
\numparts
\begin{eqnarray}
\chi_{a}(B,\xi)&=&\rho\>\lim_{t\rightarrow\infty}\int_{-\infty}^{\infty}d\eta \>\Big[F_t\left(\eta\sqrt{t},B,\xi\sqrt{t}\right)-1\Big] \label{eq:chi a 2}\\
\chi_{q}(B,\xi)&=&\rho\>\lim_{t\rightarrow\infty}\int_{-\infty}^{\infty}d\eta \>\ln F_t\left(\eta\sqrt{t},B,\xi\sqrt{t}\right) \label{eq:chi q 2}
\end{eqnarray}
\endnumparts
where we defined $\eta=y\, t^{-\frac{1}{2}}$ and $\xi=x\, t^{-\frac{1}{2}}$.

For an explicit calculation of the function $F_t(y,B,x)$ we define a diffusion propagator for a Brownian particle starting at $y$:
\begin{equation}
g(z,t \vert y,0)=\frac{1}{\sqrt{\pi \sigma^2 t}}\exp\left[ -\frac{(z-y)^2}{\sigma^2 t} \right]
\label{eq:g}
\end{equation}
where $\sigma^2/2$ is the variance of the displacement of the particle in unit time. Using this $g(z,t \vert y,0)$ we calculate the averages in \eref{eq:F} and find that $F_{t}(y,B,x)$ has a scaling form
\begin{equation}
F_t(\eta\sqrt{t},B,\xi\sqrt{t})=\cases{ f(-\eta,B,\xi) & for $\eta\le 0$\\
f(\eta,-B,-\xi) & for $\eta>0$}
\label{eq:single particle}
\end{equation}
where the scaling function $f(\eta,B,\xi)$ is given by
\begin{equation}
f(\eta,B,\xi)=1+ \frac{\left( e^{B}-1\right)}{2}\erfc{\frac{\xi+\eta}{\sigma}}
\label{eq:f}
\end{equation}
with the $\erfc{x}$ being the complementary error function
\begin{equation}
\erfc{x}=\frac{2}{\sqrt{\pi}}\int_{x}^{\infty}du\> e^{-u^2}.
\end{equation}
Substituting this in (\ref{eq:chi a 2},~\ref{eq:chi q 2}) we get
\numparts
\begin{eqnarray}
\fl \qquad \chi_{a}(B,\xi)&=&\rho \>\int_{0}^{\infty}d\eta\, \bigg\{f(\eta,B,\xi)-1\bigg\}+\rho\>\int_{0}^{\infty}d\eta\, \bigg\{f(\eta,-B,-\xi)-1\bigg\} \label{eq:chi a 3}\\
\fl \qquad \chi_{q}(B,\xi)&=&\rho\int_{0}^{\infty}d\eta\,\ln f(\eta,B,\xi)+\rho\int_{0}^{\infty}d\eta\,\ln f(\eta,-B,-\xi). \label{eq:chi q 3}
\end{eqnarray}
\endnumparts

\subsection*{Large deviation function}
\subsubsection*{Annealed:} Using the expression \eref{eq:chi a 3} into \eref{eq:ldf saddle} leads to a parametric formula of the large deviation function 
\numparts
\begin{equation}
\fl \qquad \phi_{a}(\xi)=-\rho\>\frac{\left(e^B-1\right)}{2}\int_{\xi}^{\infty}d\eta \,\erfc{\frac{\eta}{\sigma}}-\rho\>\frac{\left(e^{-B}-1\right)}{2}\int_{-\xi}^{\infty}d\eta \,\erfc{\frac{\eta}{\sigma}}
\label{eq:phia one}
\end{equation}
where the parameter $B$ is determined from the relation
\begin{equation}
\frac{\partial \phi_{a}(\xi)}{\partial B}=0.
\label{eq:dphidB}
\end{equation}
\endnumparts

In this case, the $B$ dependence can be eliminated from the formula. To show this we explicitly write the equation \eref{eq:dphidB} using the expression of $\phi_{a}(\xi)$ in \eref{eq:phia one}  which leads to
\begin{equation}
	e^{B}=\left[\int_{-\xi}^{\infty}\,d\eta\, \erfc{\frac{\eta}{\sigma}}\right]^{1/2}\left[\int_{\xi}^{\infty}\,d\eta \,\erfc{\frac{\eta}{\sigma}}\right]^{-1/2}.
\end{equation}
Substituting this in the expression \eref{eq:phia one} we get an explicit formula
\begin{equation}
	\phi_{a}(\xi)=\frac{\rho}{2}\>\left[ \sqrt{\int_{\xi}^{\infty}d\eta \,\erfc{\frac{\eta}{\sigma}}}-\sqrt{\int_{-\xi}^{\infty}d\eta \,\erfc{\frac{\eta}{\sigma}}}\,\right]^2.
	\label{eq:ldf annealed}
\end{equation}
The same result was obtained previously using the macroscopic fluctuation theory \cite{KMS2014,KMS_JSP}, and also by solving microscopic dynamics \cite{Hegde2014,KMS_JSP}. 

\subsubsection*{Quenched:} Using the expression \eref{eq:chi q 3} into \eref{eq:ldf saddle} we get
\numparts
\begin{equation}
\fl \phi_{q}(\xi)=-\rho\int_{0}^{\infty}d\eta\,\ln\left\{\left[1+ \frac{\left( e^{B}-1\right)}{2}\erfc{\frac{\eta+\xi}{\sigma}} \right]\left[1+ \frac{\left( e^{-B}-1\right)}{2}\erfc{\frac{\eta-\xi}{\sigma}} \right]\right\}
\label{eq:phi one}
\end{equation}
with the parameter $B$ determined in terms of $\xi$ by solving
\begin{equation}
	\frac{\partial \phi_{q}(\xi)}{\partial B}=0.
	\label{eq:B def}
\end{equation}
\endnumparts

Unlike the annealed case, there is no explicit formula for $\phi_q(\xi)$. We write the solution in an alternative form which will be useful later.
For this we first rewrite the expression of $\phi_q(\xi)$ in \eref{eq:phi one} as
\begin{eqnarray*}
\fl	\qquad \phi_{q}(\xi)=\rho\>\int_{0}^{\xi}d\eta ~\ln\left[\frac{1+ \left(e^{B}-1\right) \frac{1}{2}\erfc{\frac{\eta}{\sigma}}}{1+ \left(e^{-B}-1\right) \frac{1}{2}\erfc{ \frac{-\eta}{\sigma}}} \right]\nonumber\\
  -\rho\>\int_{0}^{\infty}d\eta ~\ln\left\{\Bigg[1+ \frac{\left( e^{B}-1\right)}{2}\erfc{\frac{\eta}{\sigma}}\Bigg]\Bigg[1+ \frac{\left( e^{-B}-1\right)}{2}\erfc{\frac{\eta}{\sigma}}\Bigg] \right\}.
\end{eqnarray*}
We simplify the expression using an identity  $\erfc{x}+\erfc{-x}=2$ which leads to
\numparts
\begin{equation}
		\phi_{q}(\xi)=\rho\>B~\xi- \rho\>\sigma\>\int_{0}^{\infty}d\eta \ln\Bigg[1+\sinh^2\left(\frac{B}{2}\right)\erfc{\eta}\erfc{-\eta} \Bigg]
	\label{eq:phi two}
\end{equation}
where we made a change of variable $\eta\rightarrow \sigma\>\eta$.
With $B$ given by \eref{eq:B def} yields
\begin{equation}
	\xi=\sigma\>\frac{d}{dB}\int_{0}^{\infty}d\eta \ln\Bigg[1+\sinh^2\left(\frac{B}{2}\right)\erfc{\eta}\erfc{-\eta} \Bigg].
	\label{eq:B def two}
\end{equation}
\endnumparts
The same formula was obtained previously using the macroscopic fluctuation theory \cite{KMS2014,KMS_JSP}, and solving microscopic dynamics \cite{KMS_JSP}.

\subsection*{Comparison of the two cases \label{sec:results}}
We see that the large deviation function in both settings of initial averaging is non-Gaussian. Moreover, at large values of $\xi$, the large deviation function has different asymptotics in the two settings, reflecting the sensitivity to the initial state, even at large times. To derive these asymptotics we define 
\begin{equation}
h(\xi)=\int_{\xi}^{\infty}d\eta~\erfc{\eta}.
\label{eq:alpha}
\end{equation}
In terms of this function the formula \eref{eq:ldf annealed} becomes
\begin{equation}
	\phi_{a}(\xi)=\rho\,\frac{\sigma}{2}\>\left[ \sqrt{h\left(\xi/\sigma\right)}-\sqrt{h\left(-\xi/\sigma\right)}\,\right]^2.
	\label{eq:ldf annealed alpha}
\end{equation}
For large positive $\xi$, the $h(\xi)$ can be expanded as
\begin{equation}
\fl \qquad h(\xi)=\exp\left(-\xi^2\right)\Bigg[\frac{1}{2\sqrt{\pi}~\xi^2}+\Or(\xi^{-4})\Bigg]\qquad \textrm{and}\qquad 	h(-\xi)=2\,\xi+h(\xi).
\end{equation}
Substituting this in the explicit formula \eref{eq:ldf annealed alpha} we get
\begin{equation}
	\phi_{a}(\xi)\simeq \rho \> \vert \xi \vert \qquad \qquad \textrm{for large $\xi$}.
	\label{eq:asymp anneal}
\end{equation}

For the quenched case, the formula of $\phi_q(\xi)$ is in an implicit form (\ref{eq:phi two},\,\ref{eq:B def two}). Using the formula \eref{eq:B def two} we see that large positive $\xi$ corresponds to large positive $B$. At large $B$ the integral in \eref{eq:phi two} grows as $\frac{2}{3}B^{3/2}$. Substituting this in the parametric formula, leads to
\begin{equation}
	\phi_{q}(\xi)\simeq \frac{\rho}{3\,\sigma^2} \, \vert \xi\vert^3  \qquad \qquad \textrm{for large $\xi$}
		\label{eq:asymp quench}
\end{equation}
where for negative values of $\xi$ we used the symmetry $\phi_{q}(-\xi)=\phi_{q}(\xi)$. Both asymptotics were reported earlier in \cite{KMS2014,KMS_JSP} using a macroscopic approach.

The asymptotic behaviours \eref{eq:asymp anneal} and \eref{eq:asymp quench} can be explained by a simple picture. For the annealed case, the probability of a very large displacement $x(t)$ of the tracer is dominated by the contribution from initial configurations where the interval $(0,x(t))$ is empty. The cost for the tagged particle to move a distance $x(t)$ is then negligible compared to the cost of having the whole range $(0,x(t))$ empty in the initial condition. The probability of such an initial state is
\begin{equation}
\mathcal{P}=\lim_{dx\rightarrow 0}\bigg(1-\rho\>dx\bigg)^{\frac{x(t)}{dx}}= e^{-\rho\>x(t)}.
\end{equation}
Defining $\xi=x(t)t^{-\frac{1}{2}}$ and comparing with \eref{eq:ldform} we get \eref{eq:asymp anneal}.

For the quenched case, where the initial density profile in uniform, the probability of large $x(t)$ is dominated by the events where all particles inside the interval $(0,x(t))$ have reached $x(t)$ at time $t$. The probability of such events, for large $x(t)$, is
\begin{equation}
\widetilde{\mathcal{P}}\simeq\prod_{j=0}^{x(t)\rho}g(x(t),t\vert \rho^{-1}j,0)\sim \exp\left(-\frac{\rho \>x(t)^3}{3\,\sigma^2 t}\right)
\end{equation}
where $g(x(t),t\vert y,0)$ is the diffusion propagator \eref{eq:g}. Using $\xi=x(t)t^{-\frac{1}{2}}$ and comparing with \eref{eq:ldform} we get \eref{eq:asymp quench}.

\section{Cumulants of the position of the tracer \label{sec:cumulant one}}
The cumulant generating function is related to the large deviation function by a Legendre transformation \eref{eq:legendre}.
Using the results (\ref{eq:phi two},\,\ref{eq:B def two}) and \eref{eq:ldf annealed alpha} one can derive expressions for the cumulant generating functions in both the quenched and the annealed settings.

\subsection*{Annealed}
Taking the Legendre transformation \eref{eq:legendre} with the expression of $\phi_{a}(\xi)$ in \eref{eq:ldf annealed alpha} yields a parametric formula of the cumulant generating function
\numparts
\begin{equation}
\fl \quad M_{a}(\lambda,t)=\sqrt{t}~\frac{\sigma\rho}{4}\,\bigg(\sqrt{h(\xi)}-\sqrt{h(-\xi)}\bigg)^2\left[\erfc{\xi}\sqrt{\frac{h(-\xi)}{h(\xi)}}+\erfc{-\xi}\sqrt{\frac{h(\xi)}{h(-\xi)}}\>\right]
\label{eq:Ma one}
\end{equation}
with $\xi$ determined in terms of $\lambda$ by solving
\begin{equation}
\fl \qquad	\lambda=\frac{\rho}{2}\, \bigg(\sqrt{h(-\xi)}-\sqrt{h(\xi)}\bigg)\left[\erfc{\xi}\sqrt{\frac{1}{h(\xi)}}+\erfc{-\xi}\sqrt{\frac{1}{h(-\xi)}}\,\right]
\label{eq:Ma two}
\end{equation}
\endnumparts
where the $h(\xi)$ is defined in \eref{eq:alpha}. In deriving the above formula we made a change of variables $\xi\rightarrow\sigma \>\xi$ in \eref{eq:ldf annealed alpha} and used identities $h(-\xi)-h(\xi)=2\xi$ and $\erfc{\xi}+\erfc{-\xi}=2$. The formula is simpler but equivalent to the parametric formula reported in \cite{KMS2014,KMS_JSP,Hegde2014}.

The cumulants of the tracer position $x(t)$ can be extracted using \eref{eq:cumulant expansion}. For this we first expand the expression of $M_{a}(\lambda,t)$ in \eref{eq:Ma one} in powers of $\xi$ and then use the second formula \eref{eq:Ma two} to expand $\xi$ in powers of $\lambda$. This can be implemented in Mathematica. The first non-zero cumulants obtained this way are
\numparts
\begin{eqnarray}
\bigg\langle x(t)^2 \bigg\rangle_a &=&\frac{\sigma}{\rho\sqrt{\pi}}\sqrt{t}\simeq\frac{0.56419}{\rho}\>\sigma\sqrt{t} \label{eq:variance annealed}\\
\bigg\langle x(t)^4 \bigg\rangle_a &=&\frac{\sigma}{\rho^{3}}\left[\frac{4}{\pi}-1\right]\frac{3}{\sqrt{\pi}}\sqrt{t}\simeq\frac{0.46247}{\rho^3}\>\sigma\sqrt{t} \label{eq:variance annealed b}\\
\bigg\langle x(t)^6 \bigg\rangle_a &=&\frac{\sigma}{\rho^{5}}\left[\frac{68}{\pi^2}-\frac{30}{\pi}+3\right]\frac{15}{\sqrt{\pi}}\sqrt{t}\simeq\frac{2.88197}{\rho^5}\>\sigma\sqrt{t}. \label{eq:variance annealed c}
\end{eqnarray}
\endnumparts

\subsection*{Quenched} 
Using \eref{eq:legendre} one can show that the functions $\phi_{q}(\xi)$ and $M_q(\lambda,t)$ are related by
\begin{equation}
	\phi_{q}(\xi)=\lambda ~\xi - \frac{M_q(\lambda,t)}{\sqrt{t}} \qquad \textrm{with}\qquad \xi=\frac{1}{\sqrt{t}}~\frac{\partial M_q(\lambda,t)}{\partial \lambda}.
\end{equation}
By simply comparing the above with the parametric expression (\ref{eq:phi two},\,\ref{eq:B def two}) one can get an explicit formula for the cumulant generating function
\begin{equation}
	M_{q}(\lambda,t)=\sqrt{t}\,\sigma ~\rho\int_{0}^{\infty}d\eta \ln\Bigg[1+\sinh^2\left(\frac{\lambda}{2\rho}\right)\erfc{\eta}\erfc{-\eta} \Bigg].
	\label{eq:mu one q final}
\end{equation}
This is an even function of $\lambda$ which is consistent with the fact that the microscopic dynamics is unbiased.

Cumulants are determined from $M_{q}(\lambda,t)$ using \eref{eq:cumulant expansion}. As the $M_{q}(\lambda,t)$ is an even function of $\lambda$ all the odd cumulants vanish. The first non-zero cumulants are
\numparts
\begin{eqnarray}
\fl \bigg\langle x(t)^2 \bigg\rangle_{q}&=&\frac{\sigma\sqrt{t}}{2\rho}\int_{0}^{\infty}d\eta ~ \erfc{\eta}\erfc{-\eta}=\frac{\sigma\sqrt{t}}{\rho\sqrt{2\pi}}\simeq\frac{0.39894}{\rho}\>\sigma\sqrt{t}\\
\fl \bigg\langle x(t)^4 \bigg\rangle_q&=&\frac{\sigma\sqrt{t}}{2\rho^3}\int_{0}^{\infty}d\eta ~ \left[1-\frac{3}{2}\erfc{\eta}\erfc{-\eta}\right]\erfc{\eta}\erfc{-\eta}\simeq\frac{-0.02109}{\rho^3}\>\sigma\sqrt{t}\\
\fl\bigg\langle x(t)^6 \bigg\rangle_q &=& \frac{\sigma\sqrt{t}}{2\rho^5}\int_{0}^{\infty}d\eta ~ \left\{1-\frac{15}{2}\bigg[1- \erfc{\eta}\erfc{-\eta}\bigg]\erfc{\eta}\erfc{-\eta}\right\}\erfc{\eta}\erfc{-\eta} \nonumber \\
\fl && \simeq \frac{0.00893}{\rho^5}\>\sigma\sqrt{t}.
\end{eqnarray}
\endnumparts

These expressions differ from their corresponding values \eref{eq:variance annealed}-\eref{eq:variance annealed c} in the annealed case. In particular the variances differ by a factor $\sqrt{2}$. This factor $\sqrt{2}$ between the two variances is in fact present in more general single file systems \cite{KMS2014,KMS_JSP}.

\begin{remark}
The expression \eref{eq:mu one q final} is very similar to the cumulant generating function of the time-integrated current passed through the origin \cite{Gerschenfeld2009}. In fact, one can show that for a \textit{general} single file problem in the quenched setting there is always a simple relation between these two generating functions.
We discuss this relation in \ref{app:current}. There is no such relation in the annealed setting.

As the generating function of the current is known  \cite{Gerschenfeld2009Bethe,Gerschenfeld2009} in the quenched case for the symmetric simple exclusion process (SSEP) at density $1/2$ we get the following expression for the cumulant generating function of the tracer position
\begin{equation}
\fl \qquad \qquad M_{q}(\lambda,t)=\frac{\sqrt{t}}{\pi\sqrt{2}}\int_{-\infty}^{\infty}dk\,\ln\bigg(1+\sinh^2(\lambda)e^{-k^2}\bigg)\qquad \textrm{at $\rho=\frac{1}{2}$ for SSEP}.
\label{eq:gen fnc for ssep}
\end{equation}
More details are presented in \ref{app:current}. One can easily derive the variance by expanding upto second order in powers of $\lambda$ and check that the expression is consistent with the known \cite{KMS2014,KMS_JSP} result of variance for SSEP.

\end{remark}

\section{Statistics of tracer positions at two times \label{sec:two time}}
Our method can be extended to calculate the joint probability distribution of the tracer position at multiple times. As a simple example we discuss the two time case where the joint probability has a large deviation form \eref{eq:ldform two}.

To calculate the large deviation function $\phi_{\alpha}(\xi(\tau_1),\xi(\tau_2))$ we  generalize (\ref{eq:chi annealed},\,\ref{eq:chi quench}) and define two time cumulant generating functions
\numparts
\begin{eqnarray}
 \chi_{a}(\bi{B},\bx)&=&\lim_{t\rightarrow \infty}\frac{1}{\sqrt{t}}\ln\left\langle e^{\sum_{j=1}^{2}B_j\left[R_{t_j}(x_j) -R^{\prime}_{t_j}(x_j)\right]}\right\rangle_{\textrm{\tiny evolution + initial}} \label{eq:chi two a}\\
 \chi_{q}(\bi{B},\bx)&=&\lim_{t\rightarrow \infty}\frac{1}{\sqrt{t}}\left\langle\ln\left\langle e^{\sum_{j=1}^{2}B_j\left[R_{t_j}(x_j) -R^{\prime}_{t_j}(x_j)\right]}\right\rangle_{\textrm{\tiny evolution }}\right\rangle_{\textrm{\tiny initial}} \label{eq:chi two q}
\end{eqnarray}
\endnumparts
where $t_{j}=\tau_{j}t$, and $x_{j}=\xi(\tau_{j})\sqrt{t}$. For convenience, we defined a shorthand notation $\bx\equiv\{\xi(\tau_1),\xi(\tau_2)\}$ and $\mathbf{B}\equiv\{B_1,B_2\}$. The $R_t(x)$ and $R^{\prime}_t(x)$ are defined in \sref{sec:the problem}. Generalising the arguments of \sref{sec:the problem} one can show that the large deviation function $\phi_{\alpha}(\xi(\tau_1),\xi(\tau_2))$ has a parametric solution analogous to \eref{eq:ldf saddle}:
\numparts
\begin{equation}
\phi_{\alpha}(\xi(\tau_1),\xi(\tau_2))=-\chi_{\alpha}(\bi{B},\bx)
\label{eq:ldf two formal 1}
\end{equation}
with the $\bi{B}$ determined in terms of $\bx$ by solving
\begin{equation}
\frac{\partial \chi_{\alpha}}{\partial B_j}=0 \qquad \qquad \textrm{for }j=1,2. 
\label{eq:ldf two formal 2}
\end{equation}
\endnumparts
Then the calculation of the large deviation function reduces to determine $\chi_{\alpha}(\bi{B},\bx)$ and use (\ref{eq:ldf two formal 1},\,\ref{eq:ldf two formal 2}).

An explicit solution requires to calculate the multi-particle averages in (\ref{eq:chi two a},\,\ref{eq:chi two q}). These averages can be expressed in terms of the contribution from a single Brownian particle. This leads to formulas analogous to (\ref{eq:chi a 2},\,\ref{eq:chi q 2}):
\numparts
\begin{eqnarray}
 \chi_{a}(\bi{B},\bx)&=&\rho \lim_{t\rightarrow \infty}\int_{-\infty}^{\infty}d\eta \bigg[ F_{t}\left(\eta\sqrt{t},\bi{B},\bx\sqrt{t}\right) -1 \bigg] \label{eq:chi two final a}\\
 \chi_{q}(\bi{B},\bx)&=&\rho \lim_{t\rightarrow \infty}\int_{-\infty}^{\infty}d\eta \ln\bigg[ F_{t}\left(\eta\sqrt{t},\bi{B},\bx\sqrt{t}\right)\bigg] \label{eq:chi two final q}
\end{eqnarray}
\endnumparts
with $F_{t}(y,\bi{B},\bi{x})$ given by
\begin{equation}
 F_{t}(y,\bi{B},\bi{x})=\cases{\bigg\langle e^{\sum_{j=1}^{2}B_j\>\Theta\left(z(t_j)-x_j\right)}\bigg\rangle_{z(0)=y} & for $y\le 0$\\
\bigg\langle e^{-\sum_{j=1}^{2}B_j\Theta\left(x_j-z(t_j)\right)}\bigg\rangle_{z(0)=y} & for $y>0$}
\label{eq:F two}
\end{equation}
where the average is over a Brownian particle $z(t)$ which started at $z(0)=y$. We used a notation $\bi{x}=\{x_1,x_2\}$.

It is straightforward to calculate the single particle averages in \eref{eq:F two} and we present only the final result. Similar to \eref{eq:single particle} the $F_{t}(y,\bi{B},\bi{x})$ has a scaling form
\begin{equation}
 F_{t}(\eta\sqrt{t},\bi{B},\bx\sqrt{t})=\cases{f(-\eta,\bi{B},\bx) & for $\eta\le 0$\\
f(\eta,-\bi{B},-\bx) & for $\eta > 0$}
\label{eq:F in h}
\end{equation}
where the right hand side does not depend on $t$ and the scaling function $f(\eta,\bi{B},\bx)$ has an expression
\begin{eqnarray}
\fl f(\eta,\bi{B},\bx)=1+\frac{\left( e^{B_1}-1\right)}{2}\erfc{\frac{\xi(\tau_1)+\eta}{\sigma\sqrt{\tau_1}}}+\frac{\left( e^{B_2}-1\right)}{2}\erfc{\frac{\xi(\tau_2)+\eta}{\sigma\sqrt{\tau_2}}}\nonumber\\
\fl \qquad \qquad +\left( e^{B_1}-1\right)\left( e^{B_2}-1\right)\>\frac{1}{\pi}\>\int_{\xi(\tau_1)}^{\infty}du_1\int_{\xi(\tau_2)}^{\infty}du_2\>\frac{\exp\left[ -\frac{(u_2-u_1)^2}{\sigma^2(\tau_2-\tau_1)}-\frac{(u_1+\eta)^2}{\sigma^2\tau_1}\right]}{\sigma^2\sqrt{\tau_1(\tau_2-\tau_1)}}.
\label{eq:h}
\end{eqnarray}
Substituting \eref{eq:F in h} in (\ref{eq:chi two final a},\,\ref{eq:chi two final q}) and then using (\ref{eq:ldf two formal 1},\,\ref{eq:ldf two formal 2}) we get a parametric solution of the large deviation function:

\subsubsection*{For annealed case:}
\numparts
\begin{eqnarray}
\fl ~~ \phi_{a}(\xi(\tau_1),\xi(\tau_2))=-\rho \int_{0}^{\infty}d\eta \bigg\{ f(\eta,\bi{B},\bx)-1\bigg\}-\rho \int_{0}^{\infty}d\eta \bigg\{ f(\eta,-\bi{B},-\bx)-1\bigg\}
\label{eq:phi a final 1}
\end{eqnarray}
with $B_{1}$ and $B_{2}$ determined by solving 
\begin{equation}
\frac{\partial\phi_{a}}{\partial B_{j}}=0 \qquad \textrm{for $j=1,2$}.
\label{eq:phi a final 2}
\end{equation}
\endnumparts
\subsubsection*{For quenched case:} 
\numparts
\begin{eqnarray}
\fl ~~ \phi_{q}(\xi(\tau_1),\xi(\tau_2))=-\rho \int_{0}^{\infty}d\eta \> \ln f(\eta,\bi{B},\bx) -\rho \int_{0}^{\infty}d\eta \> \ln f(\eta,-\bi{B},-\bx)
\label{eq:phi q final 1}
\end{eqnarray}
with $B_{1}$ and $B_{2}$ determined from relations
\begin{equation}
\frac{\partial\phi_{q}}{\partial B_{j}}=0 \qquad \textrm{for $j=1,2$}.
\label{eq:phi q final 2}
\end{equation}
\endnumparts

\subsection*{Two time correlation \label{sec:x1x2}}
The Legendre transformation of $\phi_{\alpha}(\xi(\tau_1),\xi(\tau_2))$ contains the information on correlations of all order between the tracer positions at two times. Generalizing \eref{eq:legendre} we define the Legendre transformation
\begin{equation}
\fl \qquad \qquad M_{\alpha}(\lambda_1,\tau_1t;\lambda_2,\tau_2t)=\sqrt{t}\>\max_{\xi(\tau_1),\xi(\tau_2)}\{\lambda_1\xi(\tau_1)+\lambda_2\xi(\tau_2)-\phi_{\alpha}(\xi(\tau_1),\xi(\tau_2))\}.
\label{eq:legendre two}
\end{equation}
Analogous to \eref{eq:cumulant expansion} one can calculate the multi-time cumulants of arbitrary order by using
\begin{equation}
\bigg\langle x(t_1)^k \>x(t_2)^{\ell} \bigg\rangle_{\alpha}=\frac{d^{k+\ell}M_{\alpha}(\lambda_1,t_1;\lambda_2,t_2)}{d\lambda_1^{k}\>d\lambda_2^{\ell}}\Bigg\vert_{(\lambda_{1},\lambda_2)=(0,0)}
\label{eq:multi cumulant}
\end{equation}
for non-negative integers $k$ and $\ell$. Using the solutions (\ref{eq:phi a final 1},\,\ref{eq:phi a final 2}) and (\ref{eq:phi q final 1},\,\ref{eq:phi q final 2}) one can then calculate the two time correlation $\langle x(t_1) x(t_2)\rangle_{\alpha}$ and obtain the results \eref{eq:corr}. (This only requires to calculate the above expressions upto quadratic order in the variables $\xi$ and $B$.)

\section{Probability of a trajectory of the tracer \label{sec:continuous time}}
An interesting generalization would be to calculate the probability weight of an entire trajectory of the tracer position. For the Brownian motion this probability can be written as $\exp(-A[X])$ where $A[X]$ is a functional of the Brownian trajectory $X(t)$ which has the form of the classical action of a free particle \cite{Kac1949,Satya2005}. We want to see how to generalize this action for the single file problem. The probability of a trajectory of the tracer has an exponential form as given in \eref{eq:ldform multi}. One can try to calculate the large deviation functional $\phi_{\alpha}[\xi]$ by generalizing the method described in \sref{sec:two time}.

We begin with a generalization of \eref{eq:F two} and define a functional
\begin{equation}
 F_t[y,B,x]=\cases{\bigg\langle e^{\int_{0}^{t}dt^{\prime} B(t^{\prime})\>\Theta\left(z(t^{\prime})-x(t^{\prime})\right)}\bigg\rangle_{z(0)=y} & for $y\le 0$\\
\bigg\langle e^{-\int_{0}^{t}dt^{\prime} B(t^{\prime})\>\Theta\left(x(t^{\prime})-z(t^{\prime})\right)}\bigg\rangle_{z(0)=y} & for $y>0$}
\label{eq:F multi}
\end{equation}
where $z(t^{\prime})$ is the trajectory of a Brownian particle in a time window $[0,t]$.

This quantity $F_t[y,B,x]$ has a simple scaling with time as in \eref{eq:F in h}. We define $\eta=t^{-\frac{1}{2}}y$, $\xi(\tau)=t^{-\frac{1}{2}}x(\tau t)$, and $b(\tau)=t\,B(\tau t)$ and denote $0\le \tau \le 1$ as the rescaled time. In terms of these variables we find a scaling functional $f[\eta,b,\xi]$ such that
\begin{equation}
 F_t\left[y,B,x\right]=\cases{f\left[-\eta,b,\xi\right] & for $\eta\le 0$\\
f[\eta,-b,-\xi] & for $\eta > 0$.}
\label{eq:F in h multi}
\end{equation}

Generalizing the line of arguments in the two time case we find that $\phi_{\alpha}[\xi]$ has a parametric solution in terms of $f[\eta,b,\xi]$:
\begin{equation}
\fl \textrm{\textit{Annealed}:} \qquad \phi_{a}\left[\xi \right]=-\rho \int_{0}^{\infty}d\eta \bigg\{ f[\eta,b,\xi]-1\bigg\}-\rho \int_{0}^{\infty}d\eta \bigg\{ f[\eta,-b,-\xi]-1\bigg\}
\label{eq:phi a final 1 multi}
\end{equation}
\begin{equation}
\fl \textrm{\textit{Quenched}:}  \qquad  \phi_{q}\left[\xi \right]=-\rho \int_{0}^{\infty}d\eta \ln f[\eta,b,\xi]-\rho \int_{0}^{\infty}d\eta \ln f[\eta,-b,-\xi]
\label{eq:phi q final 1 multi}
\end{equation}
where in both cases the corresponding functions $b(\tau)$ are determined from a functional derivative
\begin{equation}
\frac{\delta\phi_{\alpha}\left[\xi \right]}{\delta b}=0
\label{eq:phi a final 2 multi}
\end{equation}
with $\alpha$ denoting annealed or quenched.

Unlike the earlier examples, there is no closed form solution for $f[\eta,b,\xi]$. We can nevertheless write, by using the definition (\ref{eq:F multi},\,\ref{eq:F in h multi}), the solution as a series in powers of $b(\tau)$, leading to
\begin{eqnarray}
  f[\eta,b,\xi]=1+\sum_{k=1}^{\infty}\int_{0}^{1}d\tau_1\cdots\int_{0}^{\tau_{k-1}}d\tau_k \,b(\tau_1)\cdots b(\tau_k) \, W_k
  \label{eq:f series}
\end{eqnarray}
where $W_k$ is given by
\begin{equation}
W_k=\int_{\xi(\tau_1)}^{\infty}du_1\cdots \int_{\xi(\tau_k)}^{\infty}du_k\, g(u_1,\tau_1\vert u_2,\tau_2)\cdots g(u_k,\tau_k\vert -\eta,0).
\label{eq:R}
\end{equation}
Substituting this solution in (\ref{eq:phi a final 1 multi},\,\ref{eq:phi q final 1 multi}) one gets the same correlations as \eref{eq:corr} at the lowest non-trivial order.

In the \ref{app:sch} we take an alternative approach and show that $f[\eta,b,\xi]$ can be expressed in terms of the solution of a Schr\"{o}dinger like equation with a time dependent moving step potential. We did not find a way of solving this equation other than a perturbation expansion which also leads to (\ref{eq:f series},\,\ref{eq:R}).

In \cite{KMS_JSM}, the problem was approached using the macroscopic fluctuation theory. It was found that the Legendre transformation of $\phi_{\alpha}[\xi]$ can be expressed formally as a minimum action of a variational problem, but no explicit solution was derived.

\section{Summary \label{sec:conclusion}}
We have described here a simple method to analyse the large time statistics of the tracer displacement in a system of Brownian point particles with hard core repulsion. Using this method we reproduced the known results for the probability of tracer position at one time and also derived new results for the joint probability distribution of tracer position at multiple times. In particular, we derived a parametric formula for the large deviation function of the tracer position at two times and used it to calculate the two time correlation \eref{eq:corr}. The results depend on the choice of averaging over the initial state, specifically the annealed and the quenched setting. 
In a further generalization of our method we took a continuous limit and obtained a formal solution of the action which characterizes probability of an entire trajectory of the tracer in terms of the solution of a time dependent Schr\"{o}dinger equation.

At large times the joint probability of tracer positions at multiple times is asymptotically a multi-variate Gaussian distribution with a covariance matrix determined by the two time correlation \eref{eq:corr}, instead of  $\langle z(t_1)z(t_2)\rangle=2D\min(t_1,t_2)$ for the Brownian motion. It would be interesting to find, for the single file problem, analogues of celebrated properties of Brownian motion, \textit{e.g.} survival probability, distribution of maximum, \textit{etc}.

Our method is simple, but strongly relies on the connection to non-interacting particles and this makes it applicable to limited systems. For example, there is no such connection for systems where inter-particle interactions are more complicated than hardcore repulsion, or when defined on a lattice (for example, symmetric simple exclusion process).

In the present work we noticed an important connection between the tracer statistics and the current statistics in the quenched case. This connection could be useful for general single file problems, in particular for the symmetric simple exclusion process where it has been challenging to calculate the large deviation function of the tracer position: since the distribution of the integrated current is known at density $1/2$ for the quenched case \cite{Gerschenfeld2009,Gerschenfeld2009Bethe}, this connection gives the full statistics of the tracer position in the quenched case for the symmetric simple exclusion process at density $1/2$.

\ack
We thank P. L. Krapivsky and Kirone Mallick for many discussions and exchange of information on the single file problem. The work of TS is supported by a junior research chair of the Philippe Meyer Institute for Theoretical Physics at Ecole Normale Sup\'{e}rieure, Paris. We thank the Galileo Galilei Institute for Theoretical Physics in Florence for excellent working conditions and the INFN for partial support.

\appendix

\section{Relation between the distribution of integrated current and the distribution of tracer position \label{app:current}}
We consider an arbitrary single file system where the tracer always starts at the origin. The initial state is quenched at a uniform density $\rho$. Let $x(t)$ be the position of the tracer particle at time $t$ and let $Q_t(y)$ denotes the time-integrated current at position $y$ up to time $t$. We define the cumulant generating functions in the quenched setting for these two quantities as
\numparts
\begin{eqnarray}
\mu(\lambda)&=&\lim_{t\rightarrow \infty}\frac{1}{\sqrt{t}}\bigg\langle\ln \bigg \langle e^{\lambda\> x(t)} \bigg \rangle_{\textrm{\tiny evolution}}\bigg\rangle_{\textrm{\tiny initial}} \label{eq:mu def}\\
\nu(\lambda,\xi)&=&\lim_{t\rightarrow \infty}\frac{1}{\sqrt{t}}\bigg\langle\ln \bigg \langle e^{\lambda\> Q_t(\xi\sqrt{t})} \bigg \rangle_{\textrm{\tiny evolution}}\bigg\rangle_{\textrm{\tiny initial}}.
\label{eq:nu def}
\end{eqnarray}
\endnumparts
In fact, if the initial state has a flat profile at density $\rho$, one finds by translation invariance that there is no $\xi$ dependence in $\nu(\lambda,\xi)$ and 
\begin{equation}
\nu(\lambda,\xi)=\nu(\lambda,0).
\label{eq:translation}
\end{equation}

\begin{proposition}
For a general single file problem we are going to show that
\begin{equation}
\mu(\lambda)=\nu\left(\frac{\lambda}{\rho},0\right).
\label{eq:app proposition one}
\end{equation}
\end{proposition}

\begin{proof}
We prove this by showing that the following large deviation functions are equal.
\numparts
\begin{eqnarray}
\phi(\xi)&=&-\lim_{t\rightarrow \infty}\frac{\ln P(x(t)=\xi\sqrt{t})}{\sqrt{t}} \label{eq:app ldfa}\\
\varphi(q,\xi)&=&-\lim_{t\rightarrow \infty}\frac{\ln P(Q_t(\xi\sqrt{t})=q\sqrt{t})}{\sqrt{t}}
\label{eq:app ldfb}
\end{eqnarray}
\endnumparts
where $P(x(t))$ is the probability of the tracer position $x(t)$ at time $t$, whereas $P(Q_t(y))$ is the probability of integrated current $Q_t(y)$ at time $t$ at position $y$. The $\sqrt{t}$ scaling is expected for a diffusive system.

The equality between \eref{eq:app ldfa} and \eref{eq:app ldfb} comes simply from the fact that for the tracer to travel a distance $x(t)$, all the particles which were initially between the origin and $x(t)$ must cross the position $x(t)$. For the quenched initial state with uniform density $\rho$ this number of particles is $\rho \vert x(t) \vert$. This leads to
\begin{equation}
	P(x(t))= P\bigg(Q_t(x(t))=\rho x(t) \bigg).
	\label{eq:app prob R}
\end{equation}
Note that for the annealed case where the initial density is fluctuating, this equality would not hold.

Using the above equality and the large deviation form (\ref{eq:app ldfa},\,\ref{eq:app ldfb}) we get
\begin{equation}
\phi(\xi)=\varphi(\xi\>\rho,\xi).
	\label{eq:app rel 2}
\end{equation}

The cumulant generating functions $\mu(\lambda)$ and $\nu(\lambda,\xi)$ are related to the corresponding large deviation functions $\phi(\xi)$ and $\varphi(q,\xi)$ by a Legendre transformation
\begin{eqnarray}
\mu(\lambda)=\max_{\xi}\{\lambda\>\xi-\phi(\xi)\} \\
\nu(\lambda,\xi)=\max_{q}\{\lambda\>q-\varphi(q,\xi)\} 
\end{eqnarray}
Taking the Legendre transformation and using the translation invariance \eref{eq:translation} we obtain the identity \eref{eq:app proposition one}.
\end{proof}

\subsection*{An example}
We now check the relation \eref{eq:app proposition one} for a system of Brownian point particles with hard core repulsion. Comparing \eref{eq:mu def} and \eref{eq:muq} we find that $\mu(\lambda)=t^{-\frac{1}{2}}M_q(\lambda,t)$ where we have calculated $M_q(\lambda,t)$ in \eref{eq:mu one q final}. 
Here we show that the cumulant generating function $\nu(\lambda,0)$ of integrated current is given by
\begin{equation}
\nu(\lambda,0)=\rho\>\sigma\int_{0}^{\infty}d\eta\>\ln\left[1+\sinh^2\left(\frac{\lambda}{2}\right)\erfc{\eta}\erfc{-\eta}\right].
\label{eq:nu Brownian}
\end{equation}
The two results are consistent with \eref{eq:app proposition one}.

To derive \eref{eq:nu Brownian} we realize that for the integrated current we can simply ignore the hard core repulsion between particles and treat them as independent particles. Following a calculation similar to the one in \sref{sec:computation} we find
\begin{equation}
\bigg \langle e^{\lambda \> Q_t(0)} \bigg \rangle_{\textrm{\tiny evolution}}=\prod_j F_{t}(y_j,\lambda,0)
\label{eq:app Qt product}
\end{equation}
where $F(y,\lambda,x)$ is the function \eref{eq:F} and $j$ denotes all particles. Using the scaling \eref{eq:single particle} it is simple to show that the cumulant generating function
\begin{equation}
 \nu(\lambda,0)=\rho\>\int_{0}^{\infty}d\eta\>\ln f(\eta,\lambda,0)+\rho\>\int_{0}^{\infty}d\eta\>\ln f(\eta,-\lambda,0).
\end{equation}
A straightforward algebra using \eref{eq:f} and the identity $\erfc{x}+\erfc{-x}=2$ leads to the result \eref{eq:nu Brownian}.

\subsection*{Cumulant generating function for SSEP}
For the symmetric simple exclusion process (SSEP) the generating function $\nu(\lambda,0)$ of the integrated current is known \cite{Gerschenfeld2009} at density $\rho=1/2$. The expression can be derived by using a relation between the cumulant generating function of integrated current in the quenched and the annealed settings \cite[Eq.~43]{Gerschenfeld2009}:
\begin{equation}
\fl \qquad \qquad \nu(\lambda,0)\equiv\nu_{quenched}(\lambda,0)=\frac{1}{\sqrt{2}}\,\nu_{annealed}(\lambda,0)\qquad \textrm{at $\rho=\frac{1}{2}$ for SSEP,}
\label{eq:nu q a reln}
\end{equation}
where the latter quantity is defined as
\begin{equation}
\nu(\lambda,0)_{annealed}=\lim_{t\rightarrow \infty}\frac{1}{\sqrt{t}}\ln \bigg \langle e^{\lambda\> Q_t(0)} \bigg \rangle_{\textrm{\tiny evolution+initial}}.
\end{equation}
For SSEP the $\nu_{annealed}(\lambda,0)$ was calculated using Bethe ansatz in \cite{Gerschenfeld2009Bethe}. At any uniform density $0\le \rho \le 1$ it has an expression
\begin{equation}
\nu_{annealed}(\lambda,0)=\frac{1}{\pi}\int_{-\infty}^{\infty}dk\ln\left(1+4\rho(1-\rho)\sinh^2\left(\frac{\lambda}{2}\right)e^{-k^2}\right).
\end{equation}
Substituting above expression in \eref{eq:nu q a reln} and using the identity \eref{eq:app proposition one} we arrive at the result \eref{eq:gen fnc for ssep}.

\section{A formal solution of $f[\eta,b,\xi]$ \label{app:sch}}
To analyse this scaled functional $f[\eta,b,\xi]$ in \eref{eq:F in h multi} we define
\begin{equation}
 G(u,\tau \vert \eta, 0)=\bigg\langle e^{\int_{0}^{\tau}d\tau_1 b(\tau_1)\>\Theta\left(z(\tau_1)-\xi(\tau_1)\right)}\bigg\rangle_{z(0)=\eta}^{z(\tau)=u}
 \label{eq:G}
\end{equation}
for $0\le \tau \le 1$, where the subscripts and the superscripts on $\langle \rangle$ denotes the fixed initial and final positions of a Brownian motion $z(\tau_1)$, respectively. 
From \eref{eq:F multi} and \eref{eq:F in h multi} one can see that 
\begin{equation}
 f[\eta,b,\xi]=\int_{-\infty}^{\infty}du \> G(u,1 \vert -\eta, 0)
 \label{eq:fG}
\end{equation}

Considering a small increment of $\tau$ in the formula \eref{eq:G} we write
\begin{eqnarray*}
\fl G(u,\tau+\epsilon\vert \eta, 0)  =\int_{-\infty}^{\infty}dw\,g(w,\epsilon\vert 0,0)G(u+w,\tau\vert \eta, 0)\bigg\langle e^{\int_{\tau}^{\tau+\epsilon}d\tau_1 b(\tau_1)\>\Theta\left(z(\tau_1)-\xi(\tau_1)\right)}\bigg\rangle_{z(\tau)=u+w}^{z(\tau+\epsilon)=u}
\end{eqnarray*}
where $g(x,t\vert y,0)$ is the diffusion propagator defined in \eref{eq:g}. 

For small $\epsilon$ we expand the right hand side in powers of $\epsilon$ and obtain
\begin{eqnarray*}
\fl  G(u,\tau+\epsilon\vert \eta, 0) = G(u,\tau \vert \eta,0)+\epsilon\bigg\{b(\tau)\Theta(u-\xi(\tau))G(u,\tau \vert \eta,0) +\frac{\sigma^2}{4}\partial_{uu}G(u,  \tau\vert \eta, 0)\bigg\}+\cdots
\label{eq:2.2}
\end{eqnarray*}
where we used  \eref{eq:g} and $g(w,0\vert 0,0)=\delta(w)$. Taking $\epsilon \rightarrow 0$ we arrive at
\begin{equation}
 \partial_{\tau}G-\frac{\sigma^2}{4}\partial_{uu}G=b(\tau)\Theta\left(u-\xi(\tau)\right)G.
 \label{eq:G sol}
\end{equation}
This has the form of an imaginary time Schr\"{o}dinger equation with a time dependent step potential on the right hand side.
The initial condition for $G$ can be determined from its definition \eref{eq:G} which leads to
\begin{equation}
G(u,0\vert \eta,0)=\delta(u-\eta).
\end{equation}
Solution of \eref{eq:G sol} combined with relation \eref{eq:fG} determines the functional $f[\eta,b,\xi]$.

It is difficult to solve \eref{eq:G sol}. For small $b(t)$ one can however write a series solution of \eref{eq:G sol} as
\begin{eqnarray}
\fl  G(u,\tau \vert \eta, 0)=g(u,\tau\vert \eta,0)+\sum_{k=1}^{\infty}\int_{0}^{\tau}d\tau_1\cdots\int_{0}^{\tau_{k-1}}d\tau_k \,b(\tau_1)\cdots b(\tau_k)\cr
\fl \qquad \qquad  \int_{\xi(\tau_1)}^{\infty}du_1\cdots \int_{\xi(\tau_k)}^{\infty}du_k \, g(u,\tau \vert u_1,\tau_1)\cdots g(u_{k-1},\tau_{k-1} \vert u_k,\tau_k)\, g(u_k,\tau_k\vert \eta,0)
\end{eqnarray}
Substituting this in \eref{eq:fG} leads to the solution \eref{eq:f series}.
\section*{References}
\bibliographystyle{iopart-num}
\bibliography{reference}

\end{document}